\documentclass{new_tlp} 

\usepackage{paper} 

\date{2018-05-27}
\title[Syntactic Conditions for Antichain Property in \cp]
  {Syntactic Conditions for Antichain Property \\ in \cplong}
\author{\vp \\ \rice}

\begin{document}

\maketitle


\begin{abstract}
  We study syntactic conditions which guarantee when
  a \cp{} (\cplong) program has antichain property:
  no answer set is a proper subset of another.
  A notable such condition is that the program's dependency graph being acyclic
  and having no directed path from one cr-rule head literal to another.
\end{abstract}




\section{Introduction}

\ap{} (\aplong) is a programming language for knowledge representation
  and reasoning \cite{b_aprolog}.
  An \ap{} program comprises rules which determine
  the sets of beliefs that a logical agent can hold.
  \ap{} relies on the stable model semantics of logic programs with negation.

\ap{} has been applied to solve problems in various fields
  \cite{b_applications}.
  For instance, a logic program was used to guide
  multiple robots to collaboratively tidy up a house.
  Also, a tourism application suggested trips based on user preferences.

\cp{} (\cplong) extends \ap{} with cr-rules \cite{b_crprolog}.
  Cr-rules apply only when regular rules alone would result in contradiction.
  Cr-rules are meant to represent rare exceptions.

\cp{} has also been utilized in several applications.
  For instance, \cp{} enables the space shuttle decision support system
  USA-Smart to find the most reasonable plans,
  even in the unlikely case of critical failures \cite{b_usa_smart}.
  Another application of \cp{} is a formal encoding of negotiation,
  which is a multi-agent planning problem with incomplete information
  and dynamic goals \cite{b_negotiation}.
  Also, \cp{} is used as the back-end of the high-level domain representation
  of an architecture for knowledge representation and reasoning in robotics
  \cite{b_architecture}.
  Yet one more application of \cp{} is the AIA architecture
  for intentional agents who observe and response to changing environments
  \cite{b_intention}.

The first \cp{} inference engine is CR-MODELS,
  which was introduced in \citeN{b_crmodels}.
  Its efficiency is sufficient for medium-size programs,
  including an application developed for NASA.
  The second \cp{} implementation is \sparc, introduced in \citeN{b_sparc}.
  It implements a type system for the language using sort definitions.

In this paper, we investigate the antichain property
  that a logic program might have: no answer set is a proper subset of another.
  Intuitively, a program is a specification for answer sets,
  which contain literals corresponding to beliefs to be held
  by an intelligent agent \cite[pages 32-33]{b_kr}.
  The formation of these answer sets adheres to some guidelines,
  including the rationality principle which tells reasoners to believe nothing
  they are not forced to believe.
  According to this principle, the antichain property is desirable:
  no logic program $\p$ should have a chain of answer sets
  $\lits_1 \setsp \lits_2 \setsp \ldots$.
  If holding just the beliefs in $\lits_1$ suffices to satisfy the specification
  $\p$, then a reasoner should believe nothing in $\lits_2 \setd \lits_1$.

All \ap{} programs have antichain property, but some \cp{} programs do not.
  We look at syntactic conditions guaranteeing that a \cp{} program
  has this desired semantic property.
  A notable such condition -- the primary achievement of this paper --
  is when the program's dependency graph is acyclic
  and has no directed path from one cr-rule head literal to another
  (Theorem \ref{th_last}).
  We will revisit a few known results in Section \ref{s_preliminaries}
  and prove some new results in Section \ref{s_results}.


\section{Preliminaries}
\label{s_preliminaries}

The complete specifications of \ap{} and \cp{}
  can be found in \citeN[Sections 2.1 \& 2.2 \& 5.5]{b_kr}.
  We also borrow some definitions from \citeN[Sections 1 \& 2]{b_hcf}.
  In this paper, we only consider finite ground \cp{} programs
  whose abductive supports are minimal wrt (with respect to) cardinality.

\subsection{Syntax}


\begin{definitions}
  An \textdef{atom} represents a boolean value.
  A \textdef{literal} is either an atom $a$ or its \textdef{\cn} $\neg a$
  (also called \emph{strong negation}).
  An \textdef{extended literal} is either a literal $\lit$
  or its \textdef{\dn} $\keynot \lit$ (also called \emph{negation as failure}).
  Literal $\lit$ \textdef{appears positive} in extended literal $\lit$
  and \textdef{appears negative} in extended literal $\keynot \lit$.
\end{definitions}

\begin{definitions}
  A \textdef{regular rule} has the form:
  \begin{gather*}
    \regRule \tag{$\rl$}
  \end{gather*}
  Each $\lit_i$ above is a literal.
  We assume $1 \le k \le m \le n$.%
  \footnote{Sometimes, $k = 0$ is allowed,
    and $\rl$ becomes a \emph{constraint}.
    But constraints can be equivalently translated to rules with $k > 0$.
    So this paper ignores constraints for simplicity.}
  When $k = m = n$, we call $\rl$ a \textdef{fact}.
\end{definitions}

\begin{definitions}
  The \textdef{head} of a rule is the set of literals (disjuncts)
  before the arrow $\when$.
  For instance, $\head\rl = \set{\lit_1, \ldots, \lit_k}$.
  If $\rls$ is a set of rules, $\head\rls = \bigcup_{\rl \in \rls} \head\rl$.
  If $k = 1$, $\rl$ is \textdef{nondisjunctive}.
  A set $\rls$ of rules is nondisjunctive if so are all rules in $\rls$.
\end{definitions}

\begin{definitions}
  The \textdef{body} of a rule comprises the extended literals after $\when$
  (\textdef{premises} of the rule).
  The \textdef{positive body} of a rule is the set of literals
  that appear positive in the body of the rule.
  For instance, $\posbody\rl = \set{\lit_{k + 1}, \ldots, \lit_m}$.
  If $m = n$, the rule is \textdef{\dnf}.
  A set $\rls$ of rules is \dnf{} if so are all rules in $\rls$.
\end{definitions}

\begin{definitions}
  Similar to a regular rule, a \textdef{cr-rule} (consistency restoring rule)%
  \footnote{Cr-rules apply only when it would be inconsistent otherwise
    (more details in the following semantics subsection).}
  has the form: $$\crRule$$
  We call $\lit_0$ a \textdef{cr-literal}.
\end{definitions}

\begin{definitions}
  An \textdef{\ap{} program} is a finite set of regular rules.
\end{definitions}

\begin{definitions}
  A \textdef{\cp{} program} $\p$ is a finite set of regular rules and cr-rules.
  The \textdef{regular subprogram} $\preg$ comprises the regular rules in $\p$.
  The \textdef{cr-subprogram} $\pcr$ comprises the cr-rules in $\p$.
\end{definitions}

\begin{definitions}
  The \textdef{application} $\appl\rl$ of a cr-rule $\rl$ is the regular rule
  obtained from $\rl$ by replacing $\whencr$ with $\when$.
  If $\rls$ is a set of cr-rules, $\appl\rls = \set{\appl\rl : \rl \in \rls}$.
\end{definitions}

\subsection{Semantics}

We now look into the formal definitions of answer sets
  and the antichain property.
  But first, a \textdef{context} is a subset of literals in a \cp{} program.
  Two literals are \textdef{complementary} if one is the \cn{} of the other.
  A context is \textdef{consistent}
  if it contains no pair of complementary literals.

\begin{convention}[Consistent Contexts]
  For simplicity, this paper assumes all contexts (mentioned in results)
  are consistent.
\end{convention}

Now, a context $\lits$ \textdef{satisfies}:
  \begin{enumerate}
    \item a literal $\lit$ if $\lit \in \lits$
    \item an extended literal $\keynot \lit$ if $\lit \notin \lits$
    \item a regular rule head $\disj$ if some $\lit_i \in \lits$
    \item a regular rule body $\conj$ if $\lits$ satisfies all extended literals
      $\lit_{k + 1}, \ldots, \keynot \lit_n$
      (we say this rule \textdef{fires} wrt $\lits$ in case of satisfaction)
    \item a regular rule $\rl$ if $\lits$ satisfies the head of $\rl$
      whenever $\lits$ satisfies the body of $\rl$
    \item an \ap{} program $\p$ if $\lits$ satisfies every rule in $\p$
  \end{enumerate}

\begin{definitions}
  Also, a literal $\lit$ is \textdef{supported} by a regular rule $\rl$
  wrt a context $\lits$ if $\rl$ fires wrt $\lits$
  and $\head\rl \seti \lits = \set\lit$.
\end{definitions}

Next, whether a context $\lits$ is an answer set
  of an \ap{} program $\p$ is defined in two steps.
  \begin{itemize}
    \item Case $\p$ is \dnf.
      \begin{definitions}
        Then $\lits$ is an \textdef{answer set} of $\p$ if:
        $\lits$ satisfies $\p$, and $\lits$ is minimal wrt set inclusion
        (no proper subset of $\lits$ satisfies $\p$).
      \end{definitions}
    \item Case $\p$ is general.
      \begin{definitions}
        The \textdef{reduct} $\reduct{\p}{\lits}$ is the \dnf{} program
        obtained from $\p$ by:
        \begin{itemize}
          \item removing all rules containing $\keynot \lit$ where literal
            $\lit \in \lits$ (since these rules do not fire wrt $\lits$), then
          \item from each remaining rule: deleting every extended literal
            containing $\keynot \lit$ (as $\lit \notin \lits$ now,
            so $\keynot \lit$ is satisfied and can be dropped
            from the premises of the rule)
        \end{itemize}
        We say $\lits$ is an \textdef{answer set} of $\p$ if
        $\lits$ is an answer set of $\reduct{\p}{\lits}$.
        When $\p$ has some answer set, we call $\p$ \textdef{consistent}.
      \end{definitions}
  \end{itemize}

Next, we define answer sets of a \cp{} program $\p$.
  \begin{itemize}
    \item First, let $\rls \sets \pcr$
      (meaning $\rls$ is a subset of cr-rules in $\p$).
      \begin{definitions}
        Then $\rls$ is an \textdef{abductive support} of $\p$ if:
        \begin{itemize}
          \item the \ap{} program $\applp{\rls}$ is consistent, and
          \item $\rls$ is minimal wrt cardinality:
            no $\rls' \sets \pcr$ exists where $\card{\rls'} < \card{\rls}$
            such that $\applp{\rls'}$ is consistent
        \end{itemize}
      \end{definitions}
    \item Then a context $\lits$ is an \textdef{answer set} of $\p$ if
      $\lits$ is an answer set of $\applp{\rls}$
      for some abductive support $\rls$ of $\p$.
  \end{itemize}

\begin{example}[Answer Sets of a \cp{} Program]
  We encode a hypothetical complexity result using the solver \sparc{}%
  \footnote{\githubSparc}
  \cite{b_sparc}:
  \framed{
\begin{lstlisting}[firstline=7, lastline=100]
p_eq_np :- sat_p.         % P = NP if SAT is in P             (regular rule)
p_eq_np :- knapsack_p.    % P = NP if Knapsack is in P
-p_eq_np :- not surprise. % P != NP unless there is a surprise
surprise :+.              % a surprise is unlikely                 (cr-rule)
sat_p | knapsack_p.       % SAT or Knapsack is in P         (hypothetically)
\end{lstlisting}}
  \sparc{} returns exactly two answer sets
  (the cr-rule must apply to make the program consistent):
  \framed{
\begin{lstlisting}[firstline=3, lastline=100]
SPARC  V2.52
program translated
{knapsack_p, p_eq_np, surprise}

{sat_p, p_eq_np, surprise}
\end{lstlisting}}
\end{example}

We continue with a few more definitions.
  \cp{} programs $\p_1$ and $\p_2$ are \textdef{equivalent} when
  $\lits$ is an answer set of $\p_1$ iff
  $\lits$ is an answer set of $\p_2$ (for every context $\lits$).
  Finally, a \cp{} program $\p$ has \textdef{antichain property} if:
  for all answer sets $\lits_1$ and $\lits_2$ of $\p$,
  we have $\lits_1 \sets \lits_2 \impl \lits_1 = \lits_2$.
  Some \cp{} programs do not have this property.

\begin{example}[A \cp{} Program Without Antichain Property]
  Consider the following program $\p$:
  \begin{gather*}
    a \when. \\
    \neg a \when \keynot b, \keynot c. \\
    b \when c. \tag{$\rl_0$} \\
    b \whencr. \tag{$\rl_1$} \\
    c \whencr. \tag{$\rl_2$}
  \end{gather*}
  Observe $\p$ has an answer set chain
  $\lits_1 = \set{a, b} \setsp \set{a, b, c} = \lits_2$.
  (The corresponding abductive supports are $\rls_1 = \set{\rl_1}$
  and $\rls_2 = \set{\rl_2}$.)
  Intuitively, the answer set chain is induced by the \quote{dependence}
  of cr-literal $b$ (from rule $\rl_1$) on cr-literal $c$ (from rule $\rl_2$)
  in rule $\rl_0$ (\quote{$b \when c.$}).
  We will show that cr-independence guarantees antichain property,
  at least for acyclic programs such as $\p$, in Theorem \ref{th_last}.
  The terms \emph{cr-independence} and \emph{acyclicity}
  will be formally defined in Subsection \ref{ss_graph}.
\end{example}

\subsection{Antichain \ap}

Every \ap{} program is known to have antichain property;
  but for completeness, we will still provide a direct proof
  by \citeN{b_personal}.%
  \footnote{We thank the third referee for pointing out that this result
    can also be obtained from \citeN[Lemmas 1 \& 2 \& 3]{b_equivalence}.}

\begin{lemma}[Reduct Inclusion]
\label{l_reduct_inclusion}
  Let $\p$ be an \ap{} program and $\lits_1$ \& $\lits_2$ be contexts.
  If $\lits_1 \sets \lits_2$, then
  $\reduct{\p}{\lits_2} \sets \reduct{\p}{\lits_1}$.
\end{lemma}
\begin{proof}
  Assume $\reduct{\p}{\lits_2}$ has an arbitrary \dnf{} rule $\rl$:
  $$\regRuleNoDefNeg$$
  The corresponding rule in $\p$ is: $$\regRule$$
  For each $i$ in $\set{m + 1, \ldots, n}$,
  we know $\lit_i \notin \lits_2$,
  so $\lit_i \notin \lits_1$ (as $\lits_1 \sets \lits_2$).
  Therefore, $\rl$ is also a rule in $\reduct{\p}{\lits_1}$.
\end{proof}

\begin{proposition}[Antichain Property of \ap{} Programs]
\label{pr_antichain_a_prolog}
  Let $\p$ be an \ap{} program and
  $\lits_1 \sets \lits_2$ be answer sets of $\p$.
  Then $\lits_1 = \lits_2$.
\end{proposition}
\begin{proof}
  Let the reducts $\p_1 = \reduct{\p}{\lits_1}$
  and $\p_2 = \reduct{\p}{\lits_2}$.
  Notice $\lits_1$ and $\lits_2$ are respectively answer sets
  of $\p_1$ and $\p_2$.
  By Lemma \ref{l_reduct_inclusion}, $\p_2 \sets \p_1$.
  Then because $\lits_1$ satisfies $\p_1$,
  we know $\lits_1$ also satisfies $\p_2$.
  Now, being an answer set, $\lits_2$ minimally satisfies $\p_2$.
  So $\lits_2 \sets \lits_1$.
  Since $\lits_1 \sets \lits_2$ (hypothesis), we have $\lits_1 = \lits_2$.
\end{proof}



\section{Results}
\label{s_results}

We will proceed with the main contributions of this paper.
  Let us start by reviewing some concepts involving dependency graphs
  of logic programs \cite{b_hcf}.

\subsection{Dependency Graphs}
\label{ss_graph}

\begin{definitions}
  In the \textdef{dependency graph} $\g\p$ of a \cp{} program $\p$:
  every vertex is a literal in $\p$, and a directed edge
  to vertex $\lit_1$ from vertex $\lit_2$ exists iff $\p$ has some rule $\rl$
  where literals $\lit_1 \in \head\rl$ and $\lit_2 \in \posbody\rl$.
  We say $\p$ is \textdef{acyclic} if $\g\p$ contains no directed cycle.
\end{definitions}

\begin{remark}[Answer Set of Acyclic \ap{} Program]
\label{r_acyclic_answer_set}
  Let $\p$ be an acyclic%
 \ap{} program and $\lits$ be a context.
  Then $\lits$ is an answer set of $\p$ iff: $\lits$ satisfies $\p$, and
  every literal in $\lits$ is supported by a rule in $\p$ wrt $\lits$
  \cite[Theorem 2.7, page 58]{b_hcf}.
\end{remark}

%


\begin{definitions}
  Now, a \textdef{\hc} in the dependency graph $\g\p$ of a \cp{} program $\p$
  is a directed cycle $\cy$ containing vertices $\lit_1 \ne \lit_2$
  such that there is a rule $\rl \in \p$
  where literals $\lit_1, \lit_2 \in \head\rl$ \cite[page 56]{b_hcf}.
  We say $\p$ is \textdef{\hcf} if $\g\p$ contains no \hc.
\end{definitions}
The class of \hcf{} programs has several convenient properties
  that we will make use of later.

\begin{definitions}
  Also, literal $\lit_1$ \textdef{depends} on literal $\lit_2$
  in a \cp{} program $\p$ if
  the dependency graph $\g\p$ has a directed path to $\lit_1$ from $\lit_2$.
\end{definitions}
The following definition formalizes an important syntactic indicator
  of antichain property.

\begin{definition}[CR-Independence]
\label{d_cr_independence}
  A \cp{} program $\p$ is called \textdef{cr-independent} if
  $\lit_1$ does not depend on $\lit_2$ for
  all cr-literals $\lit_1$ and $\lit_2$ in $\p$.
\end{definition}

\subsection{Abductive Supports}

We continue with some technical lemmas related to abductive supports in \cp.
  Surprisingly, some of the following formal proofs are quite involved
  for their intuitive claims.

%
%
%
%
%

\begin{lemma}[Satisfying Context Intersection]
\label{l_intersection_satisfies}
  Let $\p$ be a nondisjunctive \dnf{} \ap{} program.
  If contexts $\lits_1$ and $\lits_2$ satisfy $\p$,
  then context $\lits_0 = \lits_1 \seti \lits_2$ also satisfies $\p$.
\end{lemma}
\begin{proof}
  Let $\rl$ be a rule in $\p$.
  If $\rl$ does not fire wrt $\lits_0$,
  then $\rl$ is vacuously satisfied by $\lits_0$.
  Assume $\rl$ fires wrt $\lits_0$.
  Then $\rl$ also fires wrt the supersets $\lits_1$ and $\lits_2$
  (as $\rl \in \p$ is \dnf).
  So $\lits_1$ and $\lits_2$ satisfy $\head\rl = \set\lit$
  for some literal $\lit$ (recall $\rl \in \p$ is nondisjunctive).
  Thus $\lit \in \lits_1$ and $\lit \in \lits_2$.
  Hence $\lit \in \lits_0$.
  Therefore $\lits_0$ satisfies $\rl$.
\end{proof}

The following result was obtained by \citeN{b_personal}.

\begin{lemma}[Same-Head Rule Removal \& Answer Set]
\label{l_answer_set_subprogram}
  Let $\p$ be a nondisjunctive \dnf{} \ap{} program.
  Assume $\p$ has rules $\rl_1 \ne \rl_2$ such that
  $\head{\rl_1} = \head{\rl_2}$.
  Let $\p_0 = \p \setd \set{\rl_1, \rl_2}$,
  $\p_1 = \p_0 \setu \set{\rl_1}$, and $\p_2 = \p_0 \setu \set{\rl_2}$.
  If $\lits$ is an answer set of $\p$,
  then $\lits$ is also an answer set of either $\p_1$ or $\p_2$.
\end{lemma}
\begin{proof}
  To the contrary, assume $\lits$ is an answer set of neither $\p_1$
  nor $\p_2$.
  Still, $\lits$ satisfies both $\p_1$ and $\p_2$
  (as $\lits$ satisfies their superset $\p$).
  So there exist two proper subsets of $\lits$, say $\lits_1$ and $\lits_2$,
  which respectively satisfy $\p_1$ and $\p_2$ (the programs are \dnf).
  \begin{enumerate}
    \item Case 1 of 2: either $\rl_1$ fires wrt $\lits_1$, or
      $\rl_2$ fires wrt $\lits_2$.
      \WOLOG, assume the former.
      Then $\lits_1$ satisfies $\head{\rl_1} = \head{\rl_2}$.
      So $\lits_1$ also satisfies both the rule $\rl_2$ and
      the program $\p = \p_1 \setu \set{\rl_2}$.
      As an answer set, $\lits$ minimally satisfies $\p$ (\dnf).
      But $\lits_1 \setsp \lits$, contradiction.
    \item Case 2 of 2: neither $\rl_1$ fires wrt $\lits_1$,
      nor $\rl_2$ fires wrt $\lits_2$.
      So neither $\rl_1$ nor $\rl_2$ fires wrt
      $\lits_0 = \lits_1 \seti \lits_2$ (the rules are \dnf).
      Then $\lits_0$ vacuously satisfies $\rl_1$ and $\rl_2$.
      Notice $\lits_1$ and $\lits_2$ satisfy $\p_0$
      (subset of $\p_1$ and $\p_2$),
      then $\lits_0$ satisfies $\p_0$ too
      (by Lemma \ref{l_intersection_satisfies}).
      Therefore, $\lits_0$ satisfies $\p = \p_0 \setu \set{\rl_1, \rl_2}$.
      But $\lits$ is an answer set of $\p$, and $\lits_0 \setsp \lits$,
      contradiction.
  \end{enumerate}
\end{proof}

\begin{lemma}[CR-Literal Determining CR-Rule]
\label{l_diff_cr_literals}
  Let $\p$ be a nondisjunctive \cp{} program with some abductive support $\rls$.
  For all cr-rules $\rl_1$ and $\rl_2$ in $\rls$:
  if $\head{\rl_1} = \head{\rl_2}$, then $\rl_1 = \rl_2$.
\end{lemma}
\begin{proof}
  By way of contradiction, assume there exist cr-rules $\rl_1\ne \rl_2$
  in $\rls$ where $\head{\rl_1} = \head{\rl_2}$.
  Let: $\rls_1 = \rls \setd \set{\rl_2}$
    \& $\rls_2 = \rls \setd \set{\rl_1}$ be sets of cr-rules;
  $\p_1 = \applp{\rls_1}$ \& $\p_2 = \applp{\rls_2}$ be \ap{} programs;
  $\lits$ be an answer set of $\applp{\rls}$;
  and $\p_a = \reduct{(\p_1)}\lits$ \& $\p_b = \reduct{(\p_2)}\lits$
    be (\dnf) reducts.
  By Lemma \ref{l_answer_set_subprogram}, $\lits$ is an answer set of
  either $\p_a$ or $\p_b$.
  \WOLOG, assume the former.
  Then $\lits$ is an answer set of $\p_1$.
  So $\rls_1$ is another abductive support of $\p$.
  But $\card{\rls_1} < \card{\rls}$
  (recall $\rls_1 = \rls \setd \set{\rl_2}$),
  violating the minimality of abductive support $\rls$.
\end{proof}

\begin{lemma}[CR-Literal only Supported by CR-Rule Application]
\label{l_sole_supporting_rule}
  Let $\p$ be an acyclic \cp{} program having an answer set $\lits$
  with a corresponding abductive support $\rls$.
  Let cr-rule $\rl \in \rls$ where $\head\rl = \set\lit$
  for some literal $\lit$.
  Then $\appl\rl$ is the only rule in $\pappld\rls$ which
  supports $\lit$ wrt $\lits$.
\end{lemma}
\begin{proof}
  By way of contradiction, assume $\lit$ is also supported by
  a rule $\rl' \ne \appl\rl$ in $\pappl\rls$.
  Let $\rls' = \rls \setd \set{\rl}$ and $\p' = \applp{\rls'}$.
  We will prove $\lits$ is an answer set of $\p'$:
  \begin{enumerate}
    \item First, $\lits$ satisfies $\p' \sets \pappl\rls$.
    \item Next, let $\lit'$ be an arbitrary literal in $\lits$;
      we shall show $\lit'$ is supported wrt $\lits$ by some rule in $\p'$.
      Recall that $\lits$ is an answer set of $\pappl\rls$.
      Applying Remark \ref{r_acyclic_answer_set} to $\pappl\rls$,
      we deduce that $\lit'$ is supported wrt $\lits$ by some rule $\rl_0$
      in $\pappl\rls$.
      \begin{enumerate}
        \item Case 1 of 2: $\rl_0 = \appl\rl$.
          Recall $\head\rl = \set\lit$.
          Then $\lit = \lit'$.
          Notice $\rl'$ also supports $\lit' = \lit$ wrt $\lits$, and
          $\rl' \in \p'$.
        \item Case 2 of 2: $\rl_0 \ne \appl\rl$.
          Then $\rl_0 \in \p'$ by construction.
      \end{enumerate}
      In both cases, $\lit'$ is supported by some rule in $\p'$ wrt $\lits$.
  \end{enumerate}
  Now, applying Remark \ref{r_acyclic_answer_set} to $\p'$,
  we deduce that $\lits$ is an answer set of $\p'$.
  So $\p'$ is consistent, and $\rls'$ is an abductive support of $\p$.
  But $\card{\rls'} < \card{\rls}$,
  contradicting the minimality of abductive support $\rls$.
\end{proof}

Sometimes, only the head of a rule matters semantically (but not its body),
  and we can turn it into a fact for syntactic simplicity.

\begin{definition}[Factified Rule]
  For a regular rule $\rl$, let $\factify\rl$ denote
  the \textdef{factified rule} obtained from $\rl$ by dropping the body
  of $\rl$.
  If $\rls$ is a set of rules,
  define $\factify\rls = \set{\factify\rl : \rl \in \rls}$.
\end{definition}

\begin{lemma}[Factified Abductive Support Application \& Answer Set]
\label{l_factified_answer_set}
  Let $\p$ be a \cp{} program with some answer set $\lits$ and
  a corresponding abductive support $\rls$.
  Then $\lits$ is also an answer set of the \ap{} program
  $\p' = \preg \setu \factify{\appl\rls}$.
\end{lemma}
\begin{proof}
  We prove $\lits$ is a minimal context which satisfies the reduct
  $\reduct{(\p')}\lits$:
  \begin{enumerate}
    \item Let \ap{} program $\pappld\rls$.
      Recall $\lits$ is an answer set of $\pappl\rls$
      and thus satisfies the reduct $\reduct{(\pappl\rls)}\lits =
      \reduct{(\preg)}\lits \setu \reduct{(\appl\rls)}\lits$.
      Since $\rls$ is an abductive support for answer set $\lits$,
      we know $\head\rls \sets \lits$.
      Notice $\head{\factify{\appl{\rls}}} = \head{\appl{\rls}} = \head\rls$.
      Then $\lits$ satisfies
      $\reduct{(\p')}\lits = \reduct{(\preg)}{\lits} \setu \factify{\appl\rls}$.
    \item Assume some context $\lits' \sets \lits$ also satisfies
      $\reduct{(\p')}\lits$.
      Since $\factify{\appl{\rls}}$ contains only facts,
      we know $\head{\appl{\rls}} = \head{\factify{\appl{\rls}}} \sets \lits'$.
      Then $\lits'$ satisfies $\reduct{(\pappl\rls)}\lits$.
      Recall $\lits$ minimally satisfies $\reduct{(\pappl\rls)}{\lits}$,
      as $\lits$ is an answer set of $\p$.
      So $\lits \sets \lits'$.
      Therefore $\lits' = \lits$.
  \end{enumerate}
\end{proof}


\begin{lemma}[Same-Head Abductive Supports \& Answer Set Inclusion/Equality]
\label{l_same_abductive_support_heads}
  Let $\p$ be a \cp{} program with answer sets $\lits_1 \sets \lits_2$ and
  corresponding abductive supports $\rls_1$ \& $\rls_2$.
  If $\head{\rls_1} = \head{\rls_2}$, then $\lits_1 = \lits_2$.
\end{lemma}
\begin{proof}
  By Lemma \ref{l_factified_answer_set},
  $\lits_1$ and $\lits_2$ are respectively
  answer sets of $\preg \setu \factify{\appl{\rls_1}}$ and
  $\preg \setu \factify{\appl{\rls_2}}$,
  which are the same \ap{} program because $\head{\rls_1} = \head{\rls_2}$.
  By Proposition \ref{pr_antichain_a_prolog}, since $\lits_1 \sets \lits_2$,
  we have $\lits_1 = \lits_2$.
\end{proof}


\subsection{Antichain Sufficient Condition: Acyclicity \& CR-Independence}

Next, we explore some concepts related to proofs of literals,
  which were introduced in \citeN{b_hcf}.
  Then we will be ready to prove
  the primary result of the paper: Theorem \ref{th_last}.

\begin{definition}[Proof of Literal]
\label{d_lit_proof}
  Let $\p$ be an \ap{} program, $\lits$ be a context, and $\lit$ be a literal.
  A \textdef{proof} of $\lit$ wrt $\lits$ in $\p$ is a nonempty sequence
  $\pr = \rlseq{1}{n}$ of rules in $\p$ such that:
  \begin{enumerate}
    \item the head of each rule $\rl_i$ has a literal supported by $\rl_i$
      wrt $\lits$; call this sole literal $\h{\lits}{\rl_i}$
    \item $\lit = \h{\lits}{\rl_n}$
    \item $\posbody{\rl_1} = \sete$ \label{item_first_lit_proof_rule}
    \item for every rule $\rl_i$, each literal in $\posbody{\rl_i}$
      is $\h{\lits}{\rl_j}$ for some $j < i$
  \end{enumerate}
\end{definition}
In this definition, there is a caveat
  on criterion \eqref{item_first_lit_proof_rule}.
  Details follow.

\begin{note}[Non-Fact as First Rule in Proof of Literal]
\label{n_first_lit_proof_rule}
  In the original definition of proofs of literals,
  the first rule $\rl_1$ must be a fact \cite[page 57]{b_hcf}.
  However, that seems to be too strong.
  For instance, consider a \hcf{} \ap{} program $\p$ containing a sole rule:
  \begin{gather*}
    \lit \when \keynot b. \tag{$\rl_1$}
  \end{gather*}
  The only answer set is $\lits = \set{\lit}$.
  Now, every literal in an answer set of a \hcf{} program has a proof
  \cite[Lemma B.5, page 83]{b_hcf}.
  So $\lit$ has a proof wrt $\lits$ in $\p$.
  The only candidate for such a proof is $\pr = \seq{\rl_1}$.
  But $\rl_1$ is not a fact, so there is no proof of $\lit$
  according to the original definition, contradiction.
  In the adjusted Definition \ref{d_lit_proof}, $\pr$ is a proof of $\lit$,
  since $\posbody{\rl_1} = \sete$.
  Additionally, all original results in \citeN{b_hcf} seem to still hold
  under this adjusted definition.
\end{note}

We continue with proofs of literals.
  For a proof $\pr = \rlseq1n$, let $\h{\lits}{\pr}$ denote
  $\set{\h{\lits}{\rl} : \rl \in \pr}$
  and $\posbody\pr$ denote $\set{\posbody\rl : \rl \in \pr}$.
  Also, let $\prs{\lit}{\lits}{\p}$ denote the \textdef{set of all proofs}
  of a literal $\lit$ wrt a context $\lits$ in an \ap{} program $\p$.
  A proof $\pr \in \prs{\lit}{\lits}{\p}$ is called a \textdef{minimal proof}
  if $\pr$ is shortest: there is no proof $\pr' \in \prs{\lit}{\lits}{\p}$
  where $\card{\pr'} < \card\pr$.

\begin{convention}[Distinct Rules in Proof of Literal]
  Let proof $\pr = \rlseq1n \in \prs{\lit}{\lits}{\p}$.
  As usual, each $\rl_i$ is a rule, $\lit$ is a literal, $\lits$ is a context,
  and $\p$ is an \ap{} program.
  This paper assumes that the rules in $\pr$ are pairwise distinct.
  Indeed, if there were rules $\rl_i = \rl_j$ where $i < j$,
  then $\pr' = \seq{\rllist1{j - 1}, \rllist{j + 1}n} \in \prs{\lit}{\lits}{\p}$
  would readily be a shorter proof, and $\rl_j$ would be obviously redundant.
\end{convention}

\begin{lemma}[Proofs of Literals in Answer Set]
\label{l_lit_proof_wrt_answer_set}
  If $\p$ is a \hcf{} \ap{} program with an answer set $\lits$,
  then each literal in $\lits$ has a proof wrt $\lits$ in $\p$.
\end{lemma}
\begin{proof}
  This lemma follows immediately from \citeN[Theorem 2.3, page 57]{b_hcf}.
\end{proof}

Intuitively, given an answer set $\lits$ of an \ap{} program,
  there may be an order on the literals of $\lits$
  that indicates which literal can be proven before another.
  The following concepts formalize this intuition.

\begin{definitions}
  The \textdef{rank} of a literal $\lit \in \lits$ wrt an answer set $\lits$
  in a \hcf{} \ap{} $\p$ is the postive integer $\rkt{\lit}{\lits}{\p} =
  \min\set{\card{\pr} : \pr \in \prs{\lit}{\lits}{\p}}$,
  which is the length of a minimal proof.
  Note that $\rkt{\lit}{\lits}{\p}$ is well-defined,
  since proofs $\pr$ of $\lit$ wrt $\lits$ in $\p$ exist
  due to Lemma \ref{l_lit_proof_wrt_answer_set}.
\end{definitions}

\begin{definitions}
  The \textdef{ranking function} wrt an answer set $\lits$
  in a \hcf{} \ap{} program $\p$ is a function $\rkn : \lits \to \Z^+$
  where $\rku\lit = \rkt{\lit}{\lits}{\p}$ for each literal $\lit \in \lits$.
  Note that $\rku\lit$ is well-defined, as so is $\rkt{\lit}{\lits}{\p}$.
\end{definitions}

Now, we introduce a normal proof of a literal.
  A proof can be \quote{normal} in the sense that every literal $a$
  to be derived (from the head of a rule in the proof) has higher rank
  than each of its premise literals $b$
  (from the positive body of the same rule).
  Intuitively, $a$ will be derived after $b$.
  The following definition is inspired by \citeN[Theorem 2.8, page 59]{b_hcf}.

\begin{definition}[Normal Proof of Literal]
  Let: $\p$ be a \hcf{} \ap{} program with an answer set $\lits$;
  $\rkn$ be the ranking function wrt $\lits$ in $\p$; and
  $\pr$ be a proof of a literal $\lit \in \lits$ wrt $\lits$ in $\p$.
  We say $\pr$ is a \textdef{normal proof} if:
  for each rule $\rl \in \pr$ and each literal $\lit' \in \posbody\rl$,
  we have $\rku{\h{\lits}{\rl}} > \rku{\lit'}$.
\end{definition}

The following desirable property of normal proofs will be needed later.

\begin{remark}[Normal Subproofs within Normal Proofs]
\label{r_subproof}
  Let: $\p$ be a \hcf{} \ap{} program with an answer set $\lits$;
  $\rkn$ be the ranking function wrt $\lits$ in $\p$;
  $\lit$ be a literal in $\lits$;
  $\pr = \rlseq1n$ be a normal proof in $\prs\lit\lits\p$;
  and $\rl_i$ be a rule in $\pr$.
  Then $\pr_i = \rlseq{1}{i}$ is a normal proof of $\h{\lits}{\rl_i}$
  wrt $\lits$ in $\p$.
  We say $\pr_i$ is a \textdef{subproof} within $\pr$.
\end{remark}

Now, every minimal proof is a normal proof.
  But the next example justifies the need for normal proofs by showing that
  the \quote{subproof transformation} does not preserve minimality
  (as it does normality in the previous remark).

\begin{example}[A Nonminimal Subproof within a Minimal Proof]
  Consider this acyclic \ap{} program $\p$:
  \begin{gather}
    a \when b, c.   \label{r1} \\
    b \when c1x.    \label{r2} \\
    c \when c1x.    \label{r3} \\
    c1x \when c1y.  \label{r4} \\
    c1y \when.      \label{r5} \\
    c \when c2.     \label{r6} \\
    c2 \when.       \label{r7}
  \end{gather}
  The sole answer set of $\p$ is $\lits = \set{a, b, c, c1x, c1y, c2}$.
  The only minimal proofs of literal $a$ wrt $\lits$ in $\p$ are
  the two sequences of rules
  $\seq{\eqref{r5}, \eqref{r4}, \eqref{r3}, \eqref{r2}, \eqref{r1}}$
  and $\seq{\eqref{r5}, \eqref{r4}, \eqref{r2}, \eqref{r3}, \eqref{r1}}$.
  Within both of these proofs,
  the only subproof of $c$ is $\seq{\eqref{r5}, \eqref{r4}, \eqref{r3}}$,
  which is nonminimal.
  (The minimal proof of $c$ wrt $\lits$ in $\p$
  is $\seq{\eqref{r7}, \eqref{r6}}$.)
\end{example}

Now, the following long technical lemma basically says:
  if $\lits_1 \setsp \lits_2$
  are answer sets of \ap{} programs $\p_1$ and $\p_2$, then the proofs
  of literals in $\lits_2 \setd \lits_1$ contain rules in $\p_2 \setd \p_1$.

\begin{lemma}
[Answer Set Difference Literal Proven using Program Difference Rule]
\label{l_rule_in_P2}
  Let: $\p_1$ \& $\p_2$ be \hcf{} \ap{} programs
  with corresponding answer sets $\lits_1 \setsp \lits_2$;
  $\lit$ be a literal in $\lits_2 \setd \lits_1$;
  and $\pr = \rlseq1n$ be a normal proof in $\prs{\lit}{\lits_2}{\p_2}$.
  Then there exists a rule $\rl \in \pr$ such that $\rl \in \p_2 \setd \p_1$.
\end{lemma}
\begin{proof}
  \newcommand{\fset}{\set{\rku{\lit_0} : \lit_0 \in \lits_2 \setd \lits_1}}
  Let $\rkn$ be the ranking function wrt $\lits_2$ in $\p_2$.
  We employ induction on $\rku\lit$.
  \begin{itemize}
    \item Base step: $\rku\lit = \min\fset$.
      \begin{enumerate}
        \item To the contrary, assume: for every rule $\rl \in \pr$, we have
          $\rl \in \p_1 \seti \p_2$.
        \item Then $\rl_n \in \p_1$.
        \item Since $\pr$ is a normal proof of $\lit$,
          for each literal $\lit' \in \posbody{\rl_n}$,
          we have $\rku{\lit'} < \rku\lit = \min\fset$.
          So $\lit' \in \lits_1 \seti \lits_2$.
        \item Then $\rl_n$ fires wrt $\lits_1$
          (recall: $\rl_n$ fires wrt $\lits_2$, and $\lits_1 \setsp \lits_2$).
        \item As $\lits_1$ is an answer set of $\p_1$, we know
          $\lits_1$ satisfies $\head{\rl_n}$.
        \item Let $\lit'$ be a literal in $\head{\rl_n}$.
          \begin{enumerate}
            \item Case 1 of 2: $\lit' = \lit$.
              We have already assumed $\lit \in \lits_2 \setd \lits_1$.
            \item Case 2 of 2: $\lit' \ne \lit$.
              We have $\lit' \notin \lits_2$
              (as only $\lit$ is supported by $\rl_n$ wrt $\lits_2$ in $\p_2$),
              so $\lit' \notin \lits_1$.
          \end{enumerate}
          In both cases, $\lit' \notin \lits_1$.
          So $\lits_1$ does not satisfy $\head{\rl_n}$, contradiction.
      \end{enumerate}
    \item Inductive step: $\rku\lit \le \max\fset$.
      \begin{enumerate}
        \item Induction hypothesis: for each literal
          $\lit' \in \lits_2 \setd \lits_1$,
          let $\pr'$ be a normal proof in $\prs{\lit'}{\lits_2}{\p_2}$;
          if $\rku{\lit'} < \rku\lit$, then there exists a rule ${\rl} \in \pr'$
          such that ${\rl} \in \p_2 \setd \p_1$.
        \item To the contrary, assume: for every rule $\rl \in \pr$, we have
          $\rl \in \p_1 \seti \p_2$.
          \begin{enumerate}
            \item Case 1 of 2: there exists a literal
              $\lit' \in \posbody{\rl_n}$
              where $\lit' \in \lits_2 \setd \lits_1$.
              \begin{enumerate}
                \item Notice
                  $\rku{\lit'} < \rku\lit = \rku{\h{\lits_2}{\rl_n}}$.
                \item Choose some positive integer $m < n$ where
                  $\h{\lits_2}{\rl_m} = \lit'$.
                \item As $\pr = \rlseq1n$ is a normal proof
                  in $\prs{\lit}{\lits_2}{\p_2}$,
                  we know $\pr' = \rlseq1m$ is a normal subproof
                  in $\prs{\lit'}{\lits_2}{\p_2}$, by Remark \ref{r_subproof}.
                \item By the induction hypothesis, $\pr'$ contains
                  some rule $\rl' \in \p_2 \setd \p_1$.
                \item So $\pr$ also contains ${\rl'}$.
                \item But we assumed $\rl \in \p_1 \seti \p_2$
                  for every rule $\rl \in \pr$, contradiction.
              \end{enumerate}
            \item Case 2 of 2: for every literal $\lit' \in \posbody{\rl_n}$,
              we have $\lit' \in \lits_1 \seti \lits_2$.
              \begin{enumerate}
                \item Then the rule $\rl_n$ fires wrt $\lits_1$.
                \item By our assumption, $\rl_n \in \p_1$.
                \item As $\lits_1$ is an answer set of $\p_1$, we know
                  $\lits_1$ satisfies $\head{\rl_n}$.
                \item Let $\lit'$ be a literal in $\head{\rl_n}$.
                  \begin{enumerate}
                    \item Subcase 1 of 2: $\lit' = \lit$.
                      We have already assumed $\lit \in \lits_2 \setd \lits_1$.
                    \item Subcase 2 of 2: $\lit' \ne \lit$.
                      We know $\lit' \notin \lits_2$
                      (as only $\lit$ is supported
                      by $\rl_n$ wrt $\lits_2$ in $\p_2$),
                      so $\lit' \notin \lits_1$.
                  \end{enumerate}
                  In both subcases, $\lit' \notin \lits_1$.
                  Then $\lits_1$ does not satisfy $\head{\rl_n}$, contradiction.
              \end{enumerate}
          \end{enumerate}
      \end{enumerate}
  \end{itemize}
\end{proof}

\begin{remark}[Normal/Minimal Proof of Literal \& Dependence of Proven Literal]
\label{r_minimal_dependence}
  Let proof $\pr = \rlseq1n \in \prs{\lit}{\lits}{\p}$ for some literal $\lit$
  in an answer set $\lits$ of an \ap{} program $\p$.
  If $\pr$ is a normal proof (or more specifically, a minimal proof),
  then $\lit$ depends on $\h{\lits}{\rl_i}$ for all $i < n$.
\end{remark}

The following lemma asserts (equivalently)
  that cr-independence implies antichain property in certain cases.

\begin{lemma}[Answer Set Chain Implying CR-Dependence]
\label{l_cr_dependence}
  Let $\p$ be a nondisjunctive acyclic \cp{} program.
  If $\p$ has answer sets $\lits_1 \setsp \lits_2$,
  then there exist literals $\lit_1$ and $\lit_2$ in $\head\pcr$ such that
  $\lit_1$ depends on $\lit_2$.
\end{lemma}
\begin{proof}
  Some notations first:
  \begin{enumerate}
    \item By the contrapositive of Lemma \ref{l_same_abductive_support_heads},
      there exist abductive supports $\rls_1$ and $\rls_2$
      (respectively corresponding to $\lits_1$ and $\lits_2$) where
      $\head{\rls_1} \ne \head{\rls_2}$.
    \item Construct two sets of facts: $\rls_1' = \factify{\appl{\rls_1}}$ and
      $\rls_2' = \factify{\appl{\rls_2}}$.
    \item Introduce \ap{} programs $\p_1 = \preg \setu \rls_1'$ and
      $\p_2 = \preg \setu \rls_2'$.
      By Lemma \ref{l_factified_answer_set},
      $\lits_1$ and $\lits_2$ are respectively answer sets of $\p_1$
      and $\p_2$.
  \end{enumerate}
  We follow these steps:
  \begin{enumerate}
    \item Note that $\p_1$ and $\p_2$ are nondisjunctive.
      By the contrapositive of Lemma \ref{l_diff_cr_literals},
      the cr-literals in $\rls_1$ are pairwise distinct.
      So are the cr-literals in $\rls_2$.
      Then $\card{\rls_1'} = \card{\rls_1}$ and
      $\card{\rls_2'} = \card{\rls_2}$.
    \item Notice $\card{\rls_1} = \card{\rls_2} > 0$.
      Then $\card{\rls_1'} = \card{\rls_2'} > 0$.
    \item Observe $\card{\head{\rls_1'}} = \card{\rls_1'}$ and
      $\card{\head{\rls_2'}} = \card{\rls_2'}$.
      Thus $\card{\head{\rls_1'}} = \card{\head{\rls_2'}} > 0$.
    \item Recall $\head{\rls_1'} = \head{\rls_1} \ne
      \head{\rls_2} = \head{\rls_2'}$.
      Then $\head{\rls_1'} \setd \head{\rls_2'} \ne \sete$.
    \item Select some literal
      $\lit_1 \in \head{\rls_1'} \setd \head{\rls_2'}$.
      Let $\rl_1$ be the fact \quote{$\lit_1 \when.$} in
      $\rls_1' \setd\rls_2'$.
    \item Since $\lits_1$ is an answer set of $\p_1$,
      we must have $\lit_1 \in \lits_1$.
      Recall $\lits_1 \setsp \lits_2$.
      Then $\lit_1 \in \lits_2$.
    \item As $\lits_2$ is an answer set of $\p_2$, there exists
      a rule $\rl \in \p_2$ which supports $\lit_1$ wrt $\lits_2$.
      Note that $\posbody{\rl} \sets \lits_2$.
      \begin{enumerate}
        \item Case 1 of 2: there exists a literal $\lit \in \posbody{\rl}$ where
          $\lit \in \lits_2 \setd \lits_1$.
          \begin{enumerate}
            \item Let $\pr$ be a minimal proof in $\prs{\lit}{\lits_2}{\p_2}$.
            \item By Lemma \ref{l_rule_in_P2}, there exists a rule
              $\rl_2 \in \pr$ where $\rl_2 \in \p_2 \setd \p_1$.
            \item Then $\rl_2 \in \rls_2' \setd \rls_1' \sets \appl\pcr$.
              Let literal $\lit_2 = \h{\lits_2}{\rl_2} \in \head\pcr$.
            \item As $\pr$ is a minimal proof, $\lit$ depends on $\lit_2$,
              by Remark \ref{r_minimal_dependence}.
            \item Recall $\lit_1$ depends on $\lit$ in $\rl$.
              By transitivity, $\lit_1$ depends on $\lit_2$.
          \end{enumerate}
        \item Case 2 of 2: $\posbody{\rl} \sets \lits_1 \setsp \lits_2$.
          We show that this case is impossible.
          \begin{enumerate}
            \item Subcase 1 of 2: $\rl \in \p_1 \seti \p_2$.
              \begin{enumerate}
                \item Recall $\rl$ supports $\lit_1$ wrt $\lits_2$.
                  Since $\posbody{\rl} \sets \lits_1 \setsp \lits_2$,
                  we know $\rl$ also supports $\lit_1$ wrt $\lits_1$.
                \item Applying Lemma \ref{l_sole_supporting_rule} to $\p_1$,
                  we have $\rl = \rl_1$.
                \item However, $\rl \in \p_2$
                  whereas $\rl_1 \in \p_1 \setd \p_2$, contradiction.
              \end{enumerate}
            \item Subcase 2 of 2: $\rl \in \p_2 \setd \p_1$.
              \begin{enumerate}
                \item So $\rl \in \rls_2' \setd \rls_1'$.
                  Then $\rl$ is the fact \quote{$\lit_1 \when.$},
                    which is exactly $\rl_1$.
                \item However, we selected $\rl_1$ from
                  $\rls_1' \setd \rls_2'$
                  while $\rl \in \rls_2'$, contradiction.
              \end{enumerate}
          \end{enumerate}
      \end{enumerate}
  \end{enumerate}
\end{proof}

\begin{lemma}[Equivalent Nondisjunctive Program]
\label{l_nondisjunctive_acyclic}
  For every acyclic cr-independent \cp{} program $\p$,
  there is a nondisjunctive acyclic cr-independent program $\p'$
  equivalent to $\p$.
\end{lemma}
\begin{proof}
  We will construct such a program $\p'$.
  Recall $\p = \preg \setu \pcr$,
  the union of its regular subprogram and cr-subprogram.
  Assume $\preg$ has an arbitrary rule:
  \begin{gather*}
    \regRule \tag{$\rl$}
  \end{gather*}
  We first build the nondisjunctive regular subprogram $\p_0$ of $\p'$.
  For each such rule $\rl \in \preg$,
  add the following set of $k$ rules to $\p_0$:
  \[
    \begin{rcases*}
      \lit_1  \when \conj,
        \keynot \lit_2, \keynot \lit_3, \ldots, \keynot \lit_k. \\
      \lit_2  \when \conj,
        \keynot \lit_1, \keynot \lit_3, \ldots, \keynot \lit_k. \\
      \vdots \\
      \lit_k  \when \conj,
        \keynot \lit_1, \keynot \lit_2, \ldots, \keynot \lit_{k - 1}.
    \end{rcases*}
    \tag{$R$}
  \]
  Then let $\p' = \p_0 \setu \pcr$.
  \begin{enumerate}
    \item Firstly, $\p'$ is nondisjunctive, as so are its subprograms $\p_0$
      and $\pcr$ (every cr-rule head has exactly one literal).
    \item Next, we show $\p'$ is acyclic and cr-independent.
      The only syntactic difference between $\p$ and $\p'$ is that
      $\p$ has arbitrary regular rules $\rl$
      whereas $\p'$ has corresponding collections $\rls$ of $k$ rules.
      But $\rl$ induces the same $k \cdot (m - k)$ directed edges
      as $\rls$ does.
      So the dependency graphs $\g\p = \g{\p'}$.
      Then because $\p$ is acyclic and cr-independent, so is $\p'$.
    \item Now, we prove the equivalence between $\p'$ and $\p$.
      Since $\preg$ is acyclic, it is equivalent to $\p_0$,
      by \citeN[Theorem 4.17, page 73]{b_hcf}.
      Therefore $\p = \preg \setu \pcr$ and $\p' = \p_0 \setu \pcr$
      are also equivalent.
  \end{enumerate}
\end{proof}

At last, we are ready to prove the main result of this paper.

\begin{theorem}[Antichain Property of Acyclic CR-Independent \cp{} Programs]
\label{th_last}
  If a \cp{} program $\p$ is acyclic and cr-independent,
  then $\p$ has antichain property.
\end{theorem}
\begin{proof}
  By Lemma \ref{l_nondisjunctive_acyclic}, there exists a nondisjunctive
  acyclic cr-independent program $\p'$ equivalent to $\p$.
  Now, $\p'$ has antichain property,
  by the contrapositive of Lemma \ref{l_cr_dependence}.
  Therefore, the equivalent original program $\p$ has antichain property too.
\end{proof}


\section{Conclusion}

We have found a reasonably weak syntactic condition which guarantees that
  a \cp{} program has antichain property: acyclicity and cr-independence.
  We think most natural logic programs are acyclic and cr-independent.
  In order to induce cycles, a program would need to have circular reasoning
  in some sense, which is not very helpful for practical tasks.
  Being cr-dependent is uncommon as well.
  Given that cr-rules only apply in catastrophic situations
  (when the program would be inconsistent otherwise), a natural program
  would rarely specify that a cr-literal should also be derivable
  indirectly from another cr-literal via a longer path.

The future goal is to find weaker sufficient conditions
  to extend the class of \cp{} programs known to have antichain property.
  We thank the fourth referee for the suggestion to relax Theorem \ref{th_last}
  by: either dropping acyclicity from the premises,
  or weakening it into \hcf dom.
  So far, we have found no cyclic (with or without \hc s)
  cr-independent program that has an answer set chain.
  Maybe cr-independence alone is sufficient for antichain property.
  This is a promising future research direction.


\section{Acknowledgment}

We sincerely appreciate the significant guidance of \mg{} at \ttu.
  This paper includes his proofs of Proposition \ref{pr_antichain_a_prolog} and
  Lemma \ref{l_answer_set_subprogram}.
  We also thank \eb{} for the informative discussions.
  Finally, we are indebted to the four referees for their constructive feedback.


\newpage

\bibliographystyle{acmtrans.bst}
\bibliography{crprolog.bib}


\end{document}